\documentclass[aps,prx,twocolumn,preprintnumbers,superscriptaddress,amsmath,amssymb,nofootinbib,nobalancelastpage]{revtex4-2}
\usepackage{graphicx, xcolor}
\graphicspath{{figures/}}
\usepackage{amsfonts,amsthm,microtype,mathrsfs,bbm}
\usepackage[colorlinks,allcolors=blue]{hyperref}
\usepackage{braket,mathtools,physics,enumerate}
\usepackage[capitalize]{cleveref}  
\usepackage{cases}
\usepackage{MnSymbol}
\usepackage{upgreek}

\newtheorem*{theorem*}{Theorem}
\newtheorem{proposition}{Proposition}
\newtheorem{lemma}{Lemma}
\newtheorem*{lemma*}{Lemma}

\theoremstyle{definition}
\newtheorem{definition}{Definition}
\newtheorem{example}{Example}

\newcommand{\1}{\mathbbm{1}}
\def\-{\raisebox{0 pt}{-}}
\DeclareMathOperator{\Span}{Span}
\DeclareMathOperator{\diag}{diag}

\newcommand{\prlsection}[1]{{\em {#1}---}}

\newcommand\subsetsim{\mathrel{%
\ooalign{\raise0.2ex\hbox{$\subset$}\cr\hidewidth\raise-0.8ex\hbox{\scalebox{0.9}{$\sim$}}\hidewidth\cr}}}

\begin{document}

\title{Measurement and feedforward circuits from quantum error correcting codes}

\author{Georgios Styliaris}
\thanks{Authors with equal contribution.}
\affiliation{Max Planck Institute of Quantum Optics, Hans-Kopfermann-Str. 1, Garching 85748, Germany}
\affiliation{Munich Center for Quantum Science and Technology (MCQST), Schellingstr. 4, 80799 M{\"{u}}nchen, Germany}

\author{Rahul Trivedi}
\thanks{Authors with equal contribution.}
\affiliation{Max Planck Institute of Quantum Optics, Hans-Kopfermann-Str. 1, Garching 85748, Germany}
\affiliation{Munich Center for Quantum Science and Technology (MCQST), Schellingstr. 4, 80799 M{\"{u}}nchen, Germany}

\date{\today}

\begin{abstract}
Measurements and feedforward enhance the power of shallow quantum circuits, enabling the deterministic implementation of global unitary operations and the preparation of long-range entangled states. We establish a general correspondence between all such protocols and quantum error-correcting codes: the circuit preceding the measurements acts as an encoder, and unitary feedforward can eliminate post-selection if and only if the measurement projectors are detectable errors on the codespace. This correspondence provides a common framework for state preparation and implementation of global unitaries with measurements and feedforward, turning both into a code-design problem. For stabilizer codes with Pauli measurements, although the resulting operation can be non-Clifford, we show that its nonstabilizerness originates entirely from the encoder. We overcome this restriction using non-Pauli measurements or non-additive codes, constructing protocols that generate long-range nonstabilizerness while requiring only single-qubit unitary corrections.
\end{abstract}

\maketitle
\prlsection{Introduction}What tasks can quantum computers achieve with limited resources? A natural setting to address this question is that of shallow quantum circuits~\cite{bravyi2018quantum,watts2019exponential}, since noise accumulation limits the number of layers that can be applied reliably in the absence of fault tolerance~\cite{preskill2018quantum}. The power of shallow unitary circuits with nearest-neighbor gates is, however, very limited, as correlations and entanglement can spread only over a short distance~\cite{bravyi2006lieb}. Long-range entangled states, such as GHZ states~\cite{greenberger1989bell}, metrologically useful Dicke states~\cite{pezze2018quantum}, and topologically ordered states~\cite{kitaev2003fault,levin2005string}, as well as long-range entangling gates, are thus beyond the reach of such circuits. This limitation is lifted once measurements and feedforward are allowed: measurements can spread entanglement over large distances, while feedforward compensates for unwanted measurement outcomes, enabling deterministic protocols in the spirit of quantum teleportation~\cite{bennett1993teleporting,gottesman1999demonstrating} and measurement-based quantum computation~\cite{raussendorf2001one}. Mid-circuit measurements and feedforward are also natural from an experimental point of view, as they are essential ingredients of active quantum error correction~\cite{gottesman2026surviving}. The demands of fault tolerance are expected to drive rapid improvements in both capabilities. Exploiting these advances could expand what quantum devices can achieve even before large-scale fault-tolerant quantum computation becomes available. It is therefore important to characterize the class of operations they render accessible. We refer to such protocols, which are the subject of this work, as \emph{circuit--measurement--feedforward} (CMF) protocols.

Efficient CMF protocols have been devised, both for preparing states and implementing unitary operations. This includes the preparation of tensor-network states in 1D~\cite{malz2024preparation,smith2023deterministic,smith2024constant,zhang2024characterizing,sahay2025classifying,stephen2025preparing} (including the GHZ~\cite{piroli2021quantum,watts2019exponential} and Dicke states~\cite{buhrman2024state,piroli2024approximating,yu2026efficient}), or certain subfamilies in 2D~\cite{zhang2024characterizing} and of topologically ordered states, including non-Abelian models~\cite{tantivasadakarn2023hierarchy,tantivasadakarn2024long,bravyi2022adaptive,ren2025efficient,lu2022measurement,tantivasadakarn2023shortest}. Arbitrary states can also be reached, at the cost of an exponential number of ancillas~\cite{zi2025constant}. Beyond state preparation, CMF protocols have been proposed for the implementation of long-range entangling gates~\cite{baumer2024efficient,baumer2025measurement-based,baumer2024quantum} and for the realization of dualities and gauging maps~\cite{li2023symmetry,lootens2025low,franco2025symmetry,christos2026non}. One organizing principle stands out: for Clifford circuits, time can be traded for space, as measurements and feedforward can parallelize any such circuit to constant depth using ancilla qubits~\cite{nielsen1997programmable,gottesman1999demonstrating}. For certain symmetry classes, gauging provides constructive measurement-and-feedforward protocols, including examples beyond the Clifford group~\cite{tantivasadakarn2023hierarchy,christos2026non}. These constructions leave open the more general problem of determining when measurement outcomes can be corrected by unitary feedforward for arbitrary inputs, and how this condition can be used to systematically design protocols with simple corrections.

Here we show that all CMF protocols, both for the implementation of unitary transformations and for state preparation, are in exact correspondence with \emph{quantum error correcting (QEC) codes}. The circuit prior to measurement, acting over a set of fixed ancilla qubits, can be interpreted as the \emph{encoder} of a QEC code, and the measurements as \emph{errors} over that code. The protocol then admits unitary feedforward that eliminates post-selection if and only if these errors are \emph{detectable} on the codespace. The converse also holds: any code, together with a complete family of detectable errors, gives rise to a CMF protocol, and we construct a feedforward. The correspondence thus places existing CMF protocols under a common framework and turns their design into a code-design problem. We then derive several consequences of that principle. First, we show that for stabilizer codes and Pauli-string measurements the feedforward can always be chosen to be a Pauli string while the nonlocal operation implemented by the measurement and feedforward layer is necessarily Clifford. As a result, all nonstabilizerness of the output originates from the encoder in the form of a logical operation on the code. We then design new classes of CMF protocols for which the measurements instead result in a genuinely non-Clifford action. We do so in two ways: by taking the measurements outside the Pauli group, and by considering non-additive quantum codes, while ensuring in both cases that the feedforward remains a simple product unitary. Our work thus yields systematic design principles for CMF protocols beyond the Clifford and gauging cases.

\prlsection{CMF protocols}We first formalize what we mean by a CMF protocol. Our definition treats state preparation and the implementation of unitary transformations on the same footing by referring to an \emph{isometry}, that is, a unitary acting on a register of which $n-k$ qubits are fixed to $\ket{0}$. These fixed qubits are the ancillas of the protocol and state preparation is the extreme case in which the entire input register is fixed.

\begin{definition}[CMF isometry]
    A CMF isometry
    \begin{align}
        U_{\rm targ}  \, \colon (\mathbb C^{2})^{\otimes k} \to (\mathbb C^{2})^{\otimes n} \text{ with } U_{\rm targ}^\dagger U_{\rm targ}^{} = \1_{2^{k}}
    \end{align}
    (only possible for $k \le n$) consists of
    \begin{enumerate}[(i)]
        \item a \emph{circuit}: an isometry $U_{\rm circ} \, \colon (\mathbb C^{2})^{\otimes k} \to (\mathbb C^{2})^{\otimes n}$ with $U_{\rm circ}^\dagger U_{\rm circ}^{} = \1_{2^{k}}$;
        \item \emph{measurements} on a subset of the $n$ qubits: a complete family of orthogonal projectors $\{E_a\}$ ($E_a^\dagger = E_a$, $E_a E_b = \delta_{ab} E_a$, and $\sum_a E_a = \1_{2^n}$), and
        \item \emph{feedforward}: for each outcome $a$ that occurs with nonzero probability, a unitary $V_a$ such that
    \begin{align} \label{eq:target_isometry_def}
        V_a\, E_a\, U_{\rm circ} \;\propto\; U_{\rm targ} \;, \quad \forall a \;.
    \end{align}
    \end{enumerate}
\end{definition}

By definition, a CMF isometry provides an explicit protocol to implement $U_{\rm targ}$ \emph{deterministically}. For that, one applies the circuit to the input register, with the remaining $n-k$ qubits fixed to $\ket{0}$, performs the measurement, and applies the feedforward unitary $V_a$ conditioned on the classical outcome $a$. Since \cref{eq:target_isometry_def} holds for every outcome, no post-selection is required.

Note that a CMF isometry can be post-processed by an arbitrary unitary $U_{\rm post}$, under which $U_{\rm targ} \mapsto U_{\rm post} U_{\rm targ}$ and $V_a \mapsto U_{\rm post} V_a$.
Assuming without loss of generality that the outcome labeled $a = 0$ occurs with nonzero probability, the choice $U_{\rm post} = V_0^\dagger$ renders the $a = 0$ feedforward trivial. We henceforth adopt this convention:
    \begin{align}
        V_0 = \1 \quad \text{and} \quad U_{\rm targ} \propto E_0 U_{\rm circ} \;.
    \end{align}

Our definition is phrased with a single round of measurements, but this is not restrictive, as multiple rounds can be encompassed by applying it in stages, setting $U_{\rm circ}$ of one round to be $U_{\rm targ}$ of the previous one. Our definition is thus completely general and encompasses numerous known protocols: the preparation of the GHZ state~\cite{piroli2021quantum,watts2019exponential}, of tensor-network states~\cite{smith2023deterministic,malz2024preparation,smith2024constant,sahay2025classifying,zhang2024characterizing}, and of topologically ordered states starting from symmetry-protected ones~\cite{tantivasadakarn2024long}, as well as the implementation of long-range entangling gates, such as the multi-target control-NOT~\cite{baumer2025measurement-based}, and the realization of dualities~\cite{lootens2025low}. Protocols requiring several rounds, such as the preparation of non-Abelian topological order~\cite{bravyi2022adaptive,tantivasadakarn2023hierarchy}, are covered by the multi-stage construction. As we will see, one of the consequences of this work is to show that all such protocols can be unified naturally in the language of QEC.

A CMF protocol is, in general, beneficial if it allows  to compress the depth of $U_{\rm targ}$ as compared to implementing the latter by a unitary circuit alone. All three constituents of a CMF protocol contribute to the depth: the circuit $U_{\rm circ}$, the measurement, whose projectors $E_a$ in general require a circuit to be rotated into a computational-basis measurement, and the (worst-case) feedforward $V_a$. The central problem is therefore to characterize when \cref{eq:target_isometry_def} admits a solution in which all three are shallow, while $U_{\rm targ}$ itself lies beyond the reach of shallow 2-qubit gate circuits; this needs to be separately addressed for nearest-neighbor, or some less restricted, connectivity. For now, we keep the setting fully general and we establish a general correspondence between CMF protocols and QEC codes.

\prlsection{Detectable errors}Let us recall some basic facts about QEC codes over qubits. These will be defined via an isometry, called the \emph{encoder}
\begin{align}
U_{\rm enc}: (\mathbb C^{2})^{\otimes k} \to (\mathbb C^{2})^{\otimes n}
\end{align}
encoding $k$ logical over $n$ physical qubits. The orthogonal projector onto the \emph{codespace}, defined as the image of $U_{\rm enc}$, is
    \begin{align}
            \mathcal P = U_{\rm enc}^{} U_{\rm enc}^\dagger \;.
    \end{align}
The \emph{errors} $E_a$ $(a = 0,1,\dots)$ relevant for this work will be mostly comprised of mutually orthogonal projectors
    \begin{align}
        E_a E_b = \delta_{ab} E_a, \; \quad E_a^\dagger = E_a \;.
    \end{align}
According to the Knill-Laflamme conditions~\cite{knill1997theory}, an error is called \emph{detectable} if
    \begin{align} \label{eq:detectable_def}
        \mathcal P E_a \mathcal P = c_a \mathcal P \,, \quad c_a \ge 0 \;,
    \end{align}
we call it also \emph{realizable} if $c_a \ne 0$; errors with $c_a = 0$ never occur over the codespace. To each such error we can assign an \emph{error subspace}, which is the image of
    \begin{align}
        \Pi_a = \frac{1}{c_a} E_a \mathcal P E_a \;.
    \end{align}
Note that, by \cref{eq:detectable_def}, $\Pi_a$ are orthogonal projectors and each error subspace has dimension $\tr(\Pi_a) = \Tr(\mathcal P) = 2^{k}$. Later, we will use as a reference error subspace the one with $a = 0$; recall that this outcome is assumed realizable, which is just a labeling convention. We will also be concerned with unitarily mapping error subspaces onto one another. The unitary that swaps $\Pi_0$ and $\Pi_a$ ($a \ne 0$), leaving the orthogonal complement invariant, is
    \begin{align}
        R_{a} = W^{}_a + W_a^\dagger + (\1 - \Pi_a - \Pi_0) \;, \quad
        W_a = \frac{\Pi_0 \mathcal P \Pi_a}{\sqrt{c_0c_a}} \;,
    \end{align}
as can be checked directly.

\prlsection{CMF isometries and QEC codes}Our key result is the following correspondence:

\begin{proposition}[CMF isometry $\Leftrightarrow$ QEC code]
\label{prop:CMF_QEC}$ $
    \begin{enumerate}[(i)]
        \item \emph{(CMF $\Rightarrow$ QEC.)} Every CMF isometry defines a QEC code with
            \begin{align}
                U_{\rm enc} = U_{\rm circ}
            \end{align}
        as its encoder, for which each measurement $E_a$ is a \emph{detectable error},
        \begin{align}
            \mathcal P E_a \mathcal P \propto \mathcal P\;,
        \end{align}
        where $\mathcal P$ is the codespace projector.
        \item \emph{(QEC $\Rightarrow$ CMF.)} Any QEC code with encoder $U_{\rm enc}$ that admits as detectable errors $E_a$ a complete family of mutually orthogonal projectors defines a CMF isometry with circuit
            \begin{align}
                U_{\rm circ} = U_{\rm enc} \;,
            \end{align}
        measurements the errors $E_a$, and feedforward which for $a \ne 0$ can be chosen as
        \begin{align} \label{eq:feedforward_general}
            V_a = W_a^{} + W^\dagger_a + (\1 - \Pi_0 - \Pi_a) \;.
        \end{align}
    \end{enumerate}
\end{proposition}

The proof is straightforward and can be found in the Appendix, along with the proofs of the subsequent results. Our result reduces the design of a CMF protocol to a problem of code design: one seeks a code whose encoder is a shallow circuit and which admits detectable errors that are easy to measure. It also inverts the usual role of the Knill-Laflamme conditions; here the error is induced deliberately by the measurement, and detectability is precisely what makes the protocol deterministic.

\prlsection{Stabilizer codes and Pauli measurements}The most widely used family of QEC codes is that of \emph{stabilizer} codes. Recall that a stabilizer group $\mathcal S = \langle g_1,\dots,g_{n-k} \rangle$ is an Abelian subgroup of the $n$-qubit Pauli group not containing $-\1$, generated by $n-k$ independent Pauli strings. Its \emph{stabilizer subspace} is the simultaneous $+1$ eigenspace of all generators, with projector $\mathcal P = \prod_{i} (\1 + g_i)/2$ and dimension $2^{k}$. In this formalism, a Pauli measurement is detectable error if and only if the Pauli anticommutes with at least one stabilizer generator or belongs (up to a sign) to the stabilizer group. A unitary is \emph{logical} if it preserves the codespace, $U \mathcal P U^\dagger = \mathcal P$; such a $U$ need not be Clifford.

We start with an example illustrating that, even when the QEC code in \cref{prop:CMF_QEC} is stabilizer and the measurements Pauli, the resulting CMF isometry can implement a global \emph{non-Clifford} operation:

\begin{example}[CMF fan-out] \label{ex:fanout}
    Consider the unitary
    \begin{align}
        V_{\rm fanout} = \ket{0}_1\bra{0} \otimes \bigotimes_{j=2}^{n} U^{(0)}_{j} + \ket{1}_1\bra{1} \otimes \bigotimes_{j=2}^{n} U^{(1)}_{j} \;.
    \end{align}
    Importantly, for generic single-qubit $U^{(c)}_j$, it is both a long-range entangling gate and not local-unitary equivalent to a Clifford gate (this follows directly from its nonflat operator Schmidt spectrum). Following Ref.~\cite{piroli2024approximating}, it admits the following constant-depth (i.e., $n-$independent) CMF implementation, using $n-1$ ancillas $2',\dots,n'$ initialized in $\ket{0}$ (each attached to the corresponding system site) and only nearest-neighbor gates in 1D.
    \begin{enumerate}[\emph{Step }1:]
        \item Copy the control in the computational basis, $\ket{c}_1 \mapsto \ket{c}_1 \ket{c}_{2'}\dots \ket{c}_{n'}$; this is itself a CMF protocol and requires a single round of measurements~\cite{piroli2021quantum}.
        \item Apply simultaneously the two-qubit controlled unitaries
        \begin{align}
            \bigotimes_{j=2}^{n} \left( \ket{0}_{j'}\bra{0} \otimes U^{(0)}_j + \ket{1}_{j'}\bra{1} \otimes U^{(1)}_j \right) \;.
        \end{align}
        \item Measure the Paulis $X_{2'},\dots,X_{n'}$, whereby all ancillas decouple.
        \item Apply $Z_1^{a}$, where $a \in \{0,1\}$ is the parity of the list of all measurement outcomes.
    \end{enumerate}
    We now show how $V_{\rm fanout}$, while non-Clifford, derives from a stabilizer code. Steps 1 and 2 together define an encoder $U_{\rm enc}$ onto the stabilizer subspace of $ \mathcal S = \langle  Z_1 Z_{2'}, Z_{1} Z_{3'},\dots, Z_{1} Z_{n'} \rangle$.
    It is in this logical action that the nonstabilizerness of $V_{\rm fanout}$ resides. The measurements of step 3 are detectable errors, as required by \cref{prop:CMF_QEC}.
    Finally, note that the resulting action decomposes as $V_{\rm fanout} = C \, U_{\rm enc}$, where $C$ is the Clifford unitary mapping the stabilizer generators as $Z_1 Z_{j'} \mapsto X_{j'}$.
\end{example}

We now show that several features of this example generalize:

\begin{proposition}[CMF isometries from stabilizer codes and Pauli measurements] \label{prop_stabilizer}
    Let
    \begin{align}
        \mathcal S = \langle g_1,\dots,g_{n-k} \rangle , \,
     \mathcal T = \langle P_1,\dots,P_M\rangle \; (M \le n-k)
    \end{align}
    be $n$-qubit stabilizer groups with codespace projectors $\mathcal P_{\mathcal S}$ and $\mathcal P_{\mathcal T}$. Assume that every element in $\mathcal T$ either belongs to $\mathcal S$ or anticommutes with at least one of its generators.

    Form a CMF isometry by setting $U_{\rm circ}$ an encoder onto $\mathcal P_{\mathcal S}$ and choose as measurements the Pauli strings $P_1,\dots,P_M$. Then:
        \begin{enumerate}[(i)]
        \item $U_{\rm targ} \propto \mathcal P_{\mathcal T} U_{\rm circ} $ is a CMF isometry with feedforward $V_{\boldsymbol a}$ that can be chosen to be a Pauli string.
        \item $U_{\rm targ}$ can be expressed as
    \begin{align} \label{eq:U_targ_clifford}
        U_{\rm targ} = C U_{\rm circ }
    \end{align}
    where $C$ is a Clifford unitary.
    \end{enumerate}
\end{proposition}

\noindent The proof is constructive, giving an explicit and efficient recipe for extracting the Pauli strings $V_{\boldsymbol a}$.

This result has two implications. First, it gives a systematic way to design CMF isometries with a nontrivial global action while keeping the feedforward a simple Pauli string. It also provides a stabilizer-code interpretation of numerous previous constructions, including protocols whose connection to stabilizer codes is less apparent because their encoders incorporate non-Clifford logical unitaries~\cite{piroli2024approximating,tantivasadakarn2024long,cao2026measurement}. Second, it delimits the reach of this class: the measurement process, which is what generates the long-range correlations, cannot create long-range nonstabilizerness in itself but only redistribute the nonstabilizerness already produced by the circuit (see \cref{eq:U_targ_clifford}). Recall that a state has \emph{long-range nonstabilizerness} if no geometrically local unitary circuit of depth bounded independently of the system size can transform it into a stabilizer state~\cite{ellison2021symmetry,korbany2025long,wei2025long}. Motivated by this, we now study two families resulting in a global non-Clifford action over the codespace genuinely arising from the measurement process.

\prlsection{Stabilizer codes with non-Pauli measurements}Our first route is to allow for non-Pauli measurements
    \begin{align}
        E_{a_1\dots a_M} = \prod_{i=1}^{M} \left[ \frac{1}{2}\left(\1 + (-1)^{a_i} (\cos \theta_i P_i + \sin \theta_i Q_i ) \right) \right]\;,
    \end{align}
where $P_i,Q_i$ are Pauli strings and $\theta_i \in [0, 2\pi)$. In order to have valid projectors we need
\begin{align}
    \{ P_i,Q_i \} = 0 \,, \quad \forall i 
\end{align}
while, for the different measurements to be parallelizable,
\begin{align}
        [P_i,P_j] = [Q_i,Q_j]  = [P_i,Q_j] = 0  \,,\quad \forall i \ne j \;.
\end{align}
However, the key challenge is now to find guarantee that the feedforward is easy to implement, while the CMF isometry remains nontrivial. Our result is:

\begin{proposition}[CMF isometry from stabilizer codes and non-Pauli measurements] \label{prop_nonpauli}
    Let $\mathcal S = \langle g_1,\dots,g_{n-k} \rangle$ be an $n$-qubit stabilizer group and choose $U_{\rm circ}$ an encoder onto the corresponding subspace. Take measurements specified by $P_1,\dots,P_M,Q_1,\dots,Q_M$ as above such that every $\prod_{i=1}^M P_i^{u_i} Q_i ^{v_i}$ anticommutes with at least one element in $\mathcal S$ $(u,v \in \mathbb F_2^{M}$ with $(u,v) \ne (0,0))$. Then:
    \begin{enumerate}[(i)]
        \item $U_{\rm targ} \propto E_{\boldsymbol 0} U_{\rm circ}$ is a CMF isometry with feedforward $V_{\boldsymbol{a}}$ that can be chosen to be a Pauli string for all measurement outcomes $\boldsymbol a = (a_1,\dots,a_M)$.
        \item There exist stabilizer codes and measurements satisfying the assumptions, for which
    \begin{align}
        U_{\rm targ} \ne U_{\rm CD}\, C\, U_{\rm circ}
    \end{align}
    for every constant-depth unitary circuit $U_{\rm CD}$, composed of 2-qubit gates with arbitrary connectivity, and every Clifford unitary $C$.
    \end{enumerate}
\end{proposition}

\noindent In particular, the feedforward Pauli is always of the form $V_{\boldsymbol a} = \prod_i S_i^{a_i}$, where each $S_i \in \mathcal S$ is efficiently obtained by solving a linear system (see Appendix).

Unlike Pauli measurements [\cref{eq:U_targ_clifford}], non-Pauli measurements can themselves generate long-range nonstabilizerness, rather than only redistribute the nonstabilizerness supplied by the encoder. The following Example illustrates this:

\begin{example} \label{ex:nonpauli}
    Consider the $n-$qubit CMF isometry with:
    \begin{enumerate}[(i)]
        \item Circuit $U_{\rm circ}$ an encoder onto the subspace stabilized by $\mathcal S = \langle Z_1Z_2,X_1\dots X_n \rangle$. 
        \item Measurement $E_{a} = \left[ \1 + (-1)^a (\cos \theta X_1 + \sin \theta Z_1  )\right]/2$ with $a \in  \{0,1\}$. $P_1 = X_1$ and $Q_1 = Z_1$ satisfy the conditions of \cref{prop_nonpauli}.
        \item Following the systematic procedure from the proof of \cref{prop_nonpauli}, the feedforward unitary is $V_{1} = - Y_1Y_2X_3\dots X_n$. Indeed, $V_1$ belong in $\mathcal S$, while it anticommutes with both $X_1$ and $Z_1$; hence
        \begin{align}
            E_0 U_{\rm circ} \propto V_{1} E_1 U_{\rm circ} \,,
        \end{align}
        that is, an unwanted $a = 1$ outcome can be corrected by a Pauli string.
    \end{enumerate}
    Now, take $\ket{{\rm GHZ}_n} = (\ket{0}^n + \ket{1}^n)/\sqrt{2}$, which is stabilized by $\mathcal S$. The action of the measurement $E_0$ gives
    \begin{align}
        \ket{{\rm GHZ}}_n \overset{E_0}{\longmapsto} \cos\left(\frac{\pi}{4}-\frac{\theta}{2}\right)\ket{0}^{\otimes (n-1)}+\sin\left(\frac{\pi}{4}-\frac{\theta}{2}\right)\ket{1}^{\otimes (n-1)}  ,
    \end{align}
    where for simplicity we suppressed the measured qubit, which factorizes. This is a stabilizer state precisely when $\theta$ is an integer multiple of $\pi/2$, that is, exactly when the measurement is Pauli. Whenever $\sin\theta$ is not a dyadic rational (a number of the form $m/2^r$), the nonstabilizerness cannot be removed by a constant-depth unitary circuit composed of 2-qubit gates, even with arbitrary connectivity~\cite{wei2025long}.
    
    Note that, despite generating long-range nonstabilizerness, the CMF protocol is realizable using only nearest-neighbor gates and computational basis measurements in 1D. This is since the GHZ state (more generally, $U_{\rm circ}$) can be prepared with these resources~\cite{piroli2021quantum,watts2019exponential}.
\end{example}

\prlsection{CMF protocols from non-additive QEC codes}Our second route retains Pauli measurements but replaces the stabilizer code with a non-additive one. Such codes can be constructed by combining different sectors, arising from errors, of a stabilizer code $\mathcal C$:
\begin{align}
    \mathcal C'=\bigoplus_{a=0}^{R-1}E_a\mathcal C,
    \qquad E_0=\1.
\end{align}
Here the $E_a$ are Pauli strings (not projectors) with distinct syndromes $s(E_a) \in \{ 0,1\}^{n-k}$, indicating commutation/anticommutation with the stabilizer generators. Taking $R=2^r$ adds $r$ logical qubits, while the resulting codespace need not be a stabilizer subspace. It is straightforward to show that a Pauli error $E$ is detectable on $\mathcal C'$ if
\begin{align}
    s(E)\notin
    \{s(E_a)+s(E_b):0\leq a,b<R\},
\end{align}
or if it belongs, up to a phase, to $\mathcal S(\mathcal C)$ and commutes with every $E_a$. The physical picture is that, in the former case the error maps the code to an orthogonal subspace, while in the latter it acts trivially on the code.

To obtain simple feedforward, we specialize to CSS codes~\cite{nielsen2010quantum}, whose stabilizer generators can be chosen as products of $X$ operators or products of $Z$ operators. Let $\mathcal C_Z$ denote the set of bit strings $z$ for which $\ket{z}$ satisfies all $Z$-type stabilizer constraints. Taking $E_a=X^{e_a}$, where $X^e=\bigotimes_jX^{e_j}$, gives the enlarged code
$\mathcal C'=\bigoplus_a X^{e_a}\mathcal C$. We now give simple to check sufficient conditions for a product of single-qubit $S=\diag(1,i)$ gates to correct every measurement outcome; in the Appendix we provide a complete set of necessary and sufficient conditions.

\begin{proposition}[CMF isometries from non-additive codes]
\label{prop_nonadditive}
    Let $U_{\rm circ}$ be an encoder onto $\mathcal C'$ as above, and measure $X^{f_1},\dots,X^{f_M}$. Suppose that, for each $i=1,\dots,M$, there exists $y_i\in\{0,1,2,3\}^n$ satisfying
    \begin{align}
    y_i^T\left(z+\sum_{j=1}^{M}v_j f_j\mod 2\right) = 2v_i \pmod 4.
    \end{align}
    for every $v\in\{0,1\}^M$ and
    $z\in\bigcup_{a=0}^{R-1}(e_a+\mathcal C_Z)$.
    Then
    \begin{align}
        U_{\rm targ}
        =2^{-M/2}\prod_{i=1}^{M}(\1+X^{f_i})U_{\rm circ}
    \end{align}
    is a CMF isometry with feedforward
    \begin{align}
        V_{\boldsymbol a}=\bigotimes_jS^{y_j},
        \qquad y=\sum_{i=1}^{M}a_i y_i \pmod 4.
    \end{align}
\end{proposition}

The condition guarantees both detectability and correction by single-qubit phase gates. Importantly, the resulting action on the codespace does not need to coincide with a Clifford.

\begin{example}
\label{ex:nonadditive}
    Consider $m\geq2$ blocks of six qubits and the product code $\mathcal C' = (\mathcal C'_0)^{\otimes m}$, where
\begin{align}
    \mathcal C'_0 =
    \Span\{\ket{x}:x\in\{0,1\}^6,\,
    \mathrm{wt}(x)\in\{0,4\}\}.
\end{align}
    and $\mathrm{wt}(x)$ is the Hamming weight. Each block encodes four logical qubits. Choose $U_{\rm circ}$ as a product of block encoders and measure
    \begin{align}
        P_i=\tilde X_i\tilde X_{i+1},
        \qquad
        \tilde X_i=\prod_{j=1}^{6}X_{6(i-1)+j},
    \end{align}
    for $i=1,\dots,m-1$. For outcomes $\boldsymbol a$, choose bits $v_i$ satisfying $v_i\oplus v_{i+1}=a_i$ and apply
    \begin{align}
        V_{\boldsymbol a}
        =\bigotimes_{i=1}^{m}(S^{v_i})^{\otimes6}.
    \end{align}
    This corrects every outcome.

    The codespace $\mathcal C'$ has exactly $m$ independent Pauli stabilizers, the block parities $Z^{\otimes6}$. The output retains these stabilizers and acquires the $m-1$ independent stabilizers $P_i=\tilde X_i\tilde X_{i+1}$. Since Clifford conjugation preserves the number of independent Pauli stabilizers, the resulting isometry cannot coincide with a Clifford unitary on $\mathcal C'$. Moreover, starting from $\ket{0}^{\otimes6m}$, the normalized output is $(\ket{\tilde{+}}^{\otimes m}+\ket{\tilde{-}}^{\otimes m})/\sqrt{2}$, where $\ket{\tilde{\pm}}=(\ket{0}^{\otimes6}\pm\ket{1}^{\otimes6})/\sqrt{2}$. This is a GHZ state of the blocks in the $\ket{\tilde{\pm}}$ basis, exhibiting long-range correlations. At the same time, the encoding, measurements, and corrections admit constant-depth implementations with local gates and ancillas in 1D.
    
\end{example}

\prlsection{Outlook}We established an exact correspondence between CMF protocols and QEC codes, in which the measurements correspond to errors on the code. Detectability of these errors is necessary and sufficient for unitary feedforward to eliminate post-selection. This framework unifies state preparation and performing unitary operations to arbitrary input, and we have provided systematic constructions beyond the Clifford setting with simple feedforward. For quantum variational algorithms incorporating measurements and feedforward, a natural direction is to use our toolbox to guide the systematic construction of shallow protocols with global correlations. Another direction relates to tensor-networks. Local measurements do not increase the Schmidt rank, so CMF isometries implemented with constant-depth, geometrically local circuits, and constant feedforward layers admit an exact tensor-network description with constant bond dimension. The converse problem, systematically determining when an isometry specified as a tensor-network admits an efficient CMF protocol, remains an interesting open problem.

\prlsection{Acknowledgments}We are grateful to Ignacio Cirac for insightful discussions. The authors acknowledge funding from the Munich Center for Quantum Science and Technology (MCQST), funded by the Deutsche Forschungsgemeinschaft (DFG) under Germany’s Excellence Strategy (EXC2111-390814868). R.T. acknowledges support from the European Union’s Horizon Europe research and innovation program under grant agreement number 101221560 (ToNQS).

The authors acknowledge use of AI for parts of this project. In particular, OpenAI's language models were used by the authors for assistance with proofs of the sufficient conditions in Propositions 3 and 4 as well as for designing the examples of CMF based on non-additive codes.

\bibliography{my_refs}


\appendix
\widetext

\setcounter{equation}{0}
\setcounter{figure}{0}
\setcounter{table}{0}
\setcounter{proposition}{0}
\makeatletter
\renewcommand{\thefigure}{S\arabic{figure}}

\section{Proof of Proposition~\ref{prop_app:CMF_QEC}}

\begin{proposition}[CMF isometry $\Leftrightarrow$ QEC code] \label{prop_app:CMF_QEC}
$ $
    \begin{enumerate}[(i)]
        \item \emph{(CMF $\Rightarrow$ QEC.)} Every CMF isometry defines a QEC code with
            \begin{align}
                U_{\rm enc} = U_{\rm circ}
            \end{align}
        as its encoder, for which each measurement $E_a$ is a \emph{detectable error},
        \begin{align}
            \mathcal P E_a \mathcal P \propto \mathcal P\;,
        \end{align}
        where $\mathcal P$ is the codespace projector.
        \item \emph{(QEC $\Rightarrow$ CMF.)} Any QEC code with encoder $U_{\rm enc}$ that admits as detectable errors $E_a$ a complete family of mutually orthogonal projectors defines a CMF isometry with circuit
            \begin{align}
                U_{\rm circ} = U_{\rm enc} \;,
            \end{align}
        measurements the errors $E_a$, and feedforward which for $a \ne 0$ can be chosen as
        \begin{align} 
            V_a = W_a^{} + W^\dagger_a + (\1 - \Pi_0 - \Pi_a) \;.
        \end{align}
    \end{enumerate}
\end{proposition}

\begin{proof}
    (i) Writing the feedforward condition explicitly as
    \begin{align}
    V_a E_a U_{\rm circ} = \lambda_a U_{\rm targ}
    \end{align}
    and multiplying with the adjoint, we obtain
    \begin{align}
        U_{\rm circ}^\dagger E_a U_{\rm circ}
        = |\lambda_a|^2 \1_{2^k} \;,
    \end{align}
    where we used that $V_a$ is unitary, $E_a$ is a projector, and $U_{\rm targ}$ is an isometry. Multiplying from the left by $U_{\rm circ}$ and from the right by $U_{\rm circ}^\dagger$ gives
    \begin{align}
        \mathcal P E_a \mathcal P = |\lambda_a|^2 \mathcal P \;.
    \end{align}

    (ii) Since the error $E_0$ is detectable,
    \begin{align}
        \mathcal P E_0 \mathcal P = c_0 \mathcal P \,, \quad c_0 > 0
    \end{align}
    then
    \begin{align}
        U_{\rm targ} = \frac{1}{\sqrt{c_0}} E_0 U_{\rm enc}
    \end{align}
    is an isometry; recall, by our labeling convention, $c_0 \ne 0$. Indeed, using the detectability condition and the fact that $U_{\rm enc} = \mathcal P U_{\rm enc}$,
    \begin{align}
        U_{\rm targ}^\dagger U_{\rm targ}
        = \frac{1}{c_0} U_{\rm enc}^\dagger E_0 U_{\rm enc}
        = \1_{2^k} \;.
    \end{align}
    Outcomes with $c_a=0$ never occur and require no correction. For any measurement outcome $a$ with $c_a >0$, the prescribed feedforward implies that
    \begin{align}
        V_a\, E_a\, U_{\rm circ} \;\propto\; U_{\rm targ} \;, \quad \forall a \;.
    \end{align}
    Indeed,
    \begin{align}
        V_a E_a U_{\rm enc}
        &= W_a E_a U_{\rm enc} \\
        &= \sqrt{\frac{c_a}{c_0}} E_0 U_{\rm enc} \\
        &= \sqrt{c_a}\,U_{\rm targ} \;.
    \end{align}
    Note that the chosen feedforward in is non-unique.
\end{proof}

\section{Proof of Proposition~\ref{app_prop_stabilizer}}

Recall that a stabilizer group $\mathcal S = \langle g_1,\dots,g_{n-k} \rangle$ is an Abelian subgroup of the $n$-qubit Pauli group not containing $-\1$, generated by $n-k$ independent, mutually commuting Pauli strings. Its \emph{stabilizer subspace} is the simultaneous $+1$ eigenspace of all generators, with projector $\mathcal P = \prod_{i} (\1 + g_i)/2$ and dimension $2^{k}$. Each Pauli string $g$ is represented by its \emph{symplectic vector} $v(g) \in \mathbb F_2^{2n}$ which, by convention, we take to be \emph{row} vector. Thus multiplication of Paulis maps to addition of symplectic vectors~\cite{nielsen2010quantum}.

The generating set of stabilizer group $\mathcal S = \langle g_1,\dots,g_{n-k} \rangle$ over $n$ qubits is non-unique. Valid generating sets are described by
    \begin{align} \label{eq_app:transformation_generators}
        v(g'_i) = \sum_{j} M_{ij} v(g_j) 
    \end{align}
    where $M \in \mathbb M(n-k,\mathbb F_2)$ can be any invertible matrix. Now, consider another Abelian group $\mathcal T = \langle P_1,\dots,P_M \rangle$ generated by commuting and independent Pauli strings. The \emph{syndrome matrix} between $\mathcal S$ and $\mathcal T$, for a given choice of their generators, is defined as
    \begin{align} \label{eq:syndrome_definition}
        A_{ij} = v(P_i) \Lambda v(g_j)^T \;, \quad  \Lambda = \begin{pmatrix}
                0 & \1\\ \1 & 0
        \end{pmatrix}
    \end{align}
    which has an entry 1 if $P_i$ and $g_j$ anticommute and entry 0 otherwise. Under a change of generators \cref{eq_app:transformation_generators}, the syndrome matrix transforms as
    \begin{align}
        A' = A M^T \;.
    \end{align}

We will later need the following technical lemma.

\begin{lemma} \label{lem:change_generators}
    Let $\mathcal S = \langle g_1,\dots,g_{n-k} \rangle$ be a stabilizer group and consider $\mathcal T = \langle P_1 , \dots, P_M \rangle$ (with $M \le n-k$) generated by commuting and independent Pauli strings such that all of its elements either anticommute with at least one generator of $\mathcal S$ or belong in $\mathcal S$. Then:
    \begin{enumerate}[(i)]
        \item There exists a set of commuting and independent Pauli strings $P'_1,\dots, P'_M$ such that
    \begin{align}
        \mathcal T =  \langle P'_1, \dots, P'_s, \dots P'_M \rangle
    \end{align}
    with
    \begin{align} \label{eq:intersection_P_S}
        \langle P'_1, \dots, P'_s \rangle \cap \mathcal S = \{ \1 \}
    \end{align}
    and
    \begin{align}
        \mathcal S \cap\mathcal T = \langle P'_{s+1},\dots,P'_M \rangle \;.
    \end{align}    
    \item  For $1\le i\le s$, define the syndrome matrix $A_{ij} = v(P'_i) \Lambda v(g_j)^T$. Then the row-rank of $A$ is full.
    \item There is a new set of generators $\mathcal S = \langle g'_1,\dots,g'_{n-k} \rangle$ such that
    \begin{align}
        \{P'_i,g'_i\} = 0 \;, \quad \forall i \le s
    \end{align}
    with no other anticommuting element.
    \end{enumerate}
\end{lemma}

\begin{proof}

\emph{(i)} Note that $\mathcal S \cap\mathcal T$ is a subgroup and let 
    \begin{align}
    \mathcal S \cap\mathcal T = \langle P'_{s+1},\dots,P'_M \rangle
    \end{align}    
    be an independent set of its generators. Completing this to an independent generating set of $\mathcal T$, we can write
    \begin{align}
        \mathcal T =  \langle P'_1, \dots, P'_s, \dots P'_M \rangle
    \end{align}
    and independence of the full set gives \cref{eq:intersection_P_S}.
    
\emph{(ii)} It follows by contradiction. Linear dependence of its rows would imply that there exist $\alpha_i \in \{0,1\}$, not all zero, such that
    \begin{align}
        \sum_i \alpha_i v(P'_i) \Lambda v(g_j)^T = 0 \quad \forall j \;,
    \end{align}
    i.e., that $\prod_{i=1}^{s} {P'_i}^{\alpha_i}$ commutes with every $g_j$, and hence with all of $\mathcal S$. By hypothesis, every element of $\mathcal T$ either anticommutes with some generator of $\mathcal S$ or belongs to $\mathcal S$; since $\prod_{i=1}^{s} {P'_i}^{\alpha_i}$ commutes with all generators, $\prod_{i=1}^{s} {P'_i}^{\alpha_i} = \1$, contradicting the independence of $P'_1,\dots,P'_s$.

\emph{(iii)} Since $A$ is full row rank, we pick a matrix $M$ (\cref{eq_app:transformation_generators}) such that
    \begin{align}
        A' = \begin{pmatrix}
            \1_s & 0
        \end{pmatrix}
        \;.
    \end{align}
    This choice defines a new set of generators for $\mathcal S$ with the desired property.  
\end{proof}

We are now ready for the proof.

\begin{proposition}[CMF isometries from stabilizer codes and Pauli measurements]  \label{app_prop_stabilizer}
    Let
    \begin{align}
        \mathcal S = \langle g_1,\dots,g_{n-k} \rangle , \,
     \mathcal T = \langle P_1,\dots,P_M\rangle \; (M \le n-k)
    \end{align}
    be $n$-qubit stabilizer groups with codespace projectors $\mathcal P_{\mathcal S}$ and $\mathcal P_{\mathcal T}$. Assume that every element in $\mathcal T$ either belongs to $\mathcal S$ or anticommutes with at least one of its generators.

    Form a CMF isometry by setting $U_{\rm circ}$ an encoder onto $\mathcal P_{\mathcal S}$ and choose as measurements the Pauli strings $P_1,\dots,P_M$. Then:
        \begin{enumerate}[(i)]
        \item $U_{\rm targ} \propto \mathcal P_{\mathcal T} U_{\rm circ} $ is a CMF isometry with feedforward $V_{\boldsymbol a}$ that can be chosen to be a Pauli string.
        \item $U_{\rm targ}$ can be expressed as
    \begin{align} 
        U_{\rm targ} = C U_{\rm circ }
    \end{align}
    where $C$ is a Clifford unitary.
    \end{enumerate}
\end{proposition}
\begin{proof}
    \emph{(i)} Each measurement
        \begin{align}
           E_{a_1\dots a_M} = \prod_{i=1}^{M} \left( \frac{\1 + (-1)^{a_i} P_i}{2} \right)\;,  \quad a_i \in \{ 0,1\} 
        \end{align}
    can be expressed as a linear combination of elements of $\mathcal T$ and every element of $\mathcal T$ satisfies the detectability condition $\mathcal P_{\mathcal S} (\cdot) \mathcal P_{\mathcal S} \propto \mathcal P_{\mathcal S}$. Thus $\mathcal P_{\mathcal S} E_{\boldsymbol a} \mathcal P_{\mathcal S} \propto \mathcal P_{\mathcal S}$, i.e., each $E_{\boldsymbol a}$ is a detectable error over $\mathcal S$. By \cref{prop:CMF_QEC}, $U_{\rm targ} \propto E_{\boldsymbol 0} U_{\rm circ}$ is then a CMF isometry.

    It remains to show that the feedforward can be chosen to be a Pauli string. Using \cref{lem:change_generators}, we obtain new sets of generators $\mathcal T = \langle P'_1,\dots,P'_M \rangle$ and $\mathcal S = \langle g'_1,\dots,g'_{n-k} \rangle$ such that
    \begin{align}
        \{P'_i,g'_i\} = 0 \;, \quad \forall i \le s \;,
    \end{align}
    with no other anticommuting element, and $\mathcal S \cap \mathcal T = \langle P'_{s+1},\dots,P'_M \rangle$. Suppose the two sets of generators
    \begin{align}
        \mathcal T = \langle P_1, \dots, P_M \rangle = \langle P'_1, \dots, P'_M \rangle
    \end{align}
    are related by the change-of-basis matrix $M$ as
    \begin{align}
        P'_i = \prod_j P_{j}^{M_{ij}} \;.
    \end{align}
    Substituting we have
    \begin{align}
        E_{a_1\dots a_M} \propto \prod_i \left( \1 + (-1)^{a_i} P_i \right) = \prod_i \left( \1 + (-1)^{\sum_{j}M_{ij}a_j} P'_i \right) \;.
    \end{align}
    From \cref{lem:change_generators}, $g'_i P'_i = -P'_i g'_i$ while $g'_i$ commutes with all $P'_j$ for $j \ne i$. Note that this holds for all $i \le s$, while for $i > s$ the corresponding $P'_i$ are part of the stabilizer. As a result, we can choose
    \begin{align}
        V_{a_1\dots a_M} = \prod_{i=1}^s {g_i'}^{\sum_{j}M_{ij}a_j} \;,
    \end{align}
    a Pauli string, which indeed achieves
    \begin{align}
        V_{a_1\dots a_M} E_{a_1\dots a_M} U_{\rm circ} \propto E_{0\dots 0} U_{\rm circ} \;.
    \end{align}

    \emph{(ii)} Define
\begin{align}
    C = \prod_{i=1}^{s} \frac{\1 + P'_i g'_i}{\sqrt{2}} \;.
\end{align}
Each factor is a Clifford unitary, while the different factors commute by the
commutation relations established in (i). Since
$g'_i U_{\rm circ} = U_{\rm circ}$, these relations also give
\begin{align}
    C U_{\rm circ}
    &= 2^{-s/2}\prod_{i=1}^{s}(\1 + P'_i)U_{\rm circ} \\
    &= 2^{s/2}E_{\boldsymbol 0}U_{\rm circ} \;,
\end{align}
where we used that 
\begin{align}
    E_{\boldsymbol 0}
    = \prod_{i=1}^{M}\frac{\1 + P_i}{2}
    = \prod_{i=1}^{M}\frac{\1 + P'_i}{2} \;.
\end{align}
and that $P'_i \in \mathcal S$ for $i>s$.
Unitarity of $C$ implies $c_{\boldsymbol 0}=2^{-s}$, and hence
$U_{\rm targ}=C U_{\rm circ}$.
\end{proof}

\section{Proof of Proposition~\ref{app_prop_nonpauli}}

We recall the setting. The measurements are
    \begin{align}
        E_{a_1\dots a_M} = \prod_{i=1}^{M} \left[ \frac{1}{2}\left(\1 + (-1)^{a_i} (\cos \theta_i P_i + \sin \theta_i Q_i ) \right) \right]\;,
    \end{align}
with $P_i,Q_i$ Pauli strings satisfying
\begin{align}
    \{ P_i,Q_i \} = 0 \,, \quad \forall i 
\end{align}
so that the $E_{\boldsymbol a}$ are projectors, and
\begin{align}
        [P_i,P_j] = [Q_i,Q_j]  = [P_i,Q_j] = 0  \,,\quad \forall i \ne j \;,
\end{align}
so that the measurements are parallelizable.

\begin{proposition}[CMF isometry from stabilizer codes and non-Pauli measurements] \label{app_prop_nonpauli}
    Let $\mathcal S = \langle g_1,\dots,g_{n-k} \rangle$ be an $n$-qubit stabilizer group and choose $U_{\rm circ}$ an encoder onto the corresponding subspace. Take measurements specified by $P_1,\dots,P_M,Q_1,\dots,Q_M$ as above such that every $\prod_{i=1}^M P_i^{u_i} Q_i ^{v_i}$ anticommutes with at least one element in $\mathcal S$ $(u,v \in \mathbb F_2^{M}$ with $(u,v) \ne (0,0))$. Then:
    \begin{enumerate}[(i)]
        \item $U_{\rm targ} \propto E_{\boldsymbol 0} U_{\rm circ}$ is a CMF isometry with feedforward $V_{\boldsymbol{a}}$ that can be chosen to be a Pauli string for all measurement outcomes $\boldsymbol a = (a_1,\dots,a_M)$.
        \item There exist stabilizer codes and measurements satisfying the assumptions, for which
    \begin{align}
        U_{\rm targ} \ne U_{\rm CD}\, C\, U_{\rm circ}
    \end{align}
    for every constant-depth unitary circuit $U_{\rm CD}$, composed of 2-qubit gates with arbitrary connectivity, and every Clifford unitary $C$.
    \end{enumerate}
\end{proposition}

\begin{proof}
    \emph{(i)} We need to show that there exist Pauli strings $V_{a_1 \dots a_M}$ such that
    \begin{align}
        E_{0\dots 0} U_{\rm circ} \propto V_{a_1 \dots a_M} E_{a_1 \dots a_M} U_{\rm circ} \quad \forall a_j \in \{0,1\} \,.
    \end{align}
    To establish this, we will construct Pauli strings $S_i \in \mathcal S$ such that
    \begin{subequations} \label{subeq:comm_SPQ}
    \begin{align}
        \{ S_{i}, P_i\} &= \{ S_{i}, Q_i\} = 0 \quad \forall i \\
        [ S_i,P_j ] &= [S_i , Q_j] = 0 \quad \forall i\ne j
    \end{align}
    \end{subequations}
    which implies we can choose
    \begin{align}
        V_{a_1\dots a_M} = \prod_{i=1}^M S^{a_i}_i \;.
    \end{align}
    To construct the $S_i$, consider the combined syndrome matrix
    \begin{align}
        A_{ij} = \begin{cases}
			v(P_i) \Lambda v(g_j)^T, & 1 \le i \le M \,,\\
            v(Q_{i-M}) \Lambda v(g_j)^T, & M+1 \le i \le 2M \,.
		 \end{cases}
    \end{align}
    We now show that $A$ is full row rank. Indeed, consider two binary (row) vectors $u,v \in \mathbb F_2^{M}$ with $(u,v) \ne (0,0)$. Then
    \begin{align}
        (u,v)^T \in \ker (A^T)  \; \Longleftrightarrow \; \left[ \prod_{i=1}^M P_i^{u_i} Q_i ^{v_i} , g_j \right] = 0 \quad \forall j
    \end{align}
    which is forbidden by the hypothesis; thus $A^T$ has trivial kernel. Therefore $A$ is surjective and the equation
    \begin{align} \label{eq:w_unit_vecs}
        A w_i^T = (e_i,e_i)^T
    \end{align}
    for $w_i \in \mathbb F_2^{n-k}$ has a solution; here $e_i = (0_1,\dots,1_i,\dots,0_M)$. Defining
    \begin{align}
        S_i = \prod_{j=1}^{n-k} g_j^{(w_i)_j}
    \end{align}
    we get from \cref{eq:w_unit_vecs} the desired commutation and anticommutation relations Eqs.~\eqref{subeq:comm_SPQ} while, by construction, $S_i \in \mathcal S$ and is thus a Pauli string.

    \emph{(ii)} For $\sin\theta$ not a dyadic rational, the normalized measurement of $\ket{\mathrm{GHZ}_n}$ in \cref{ex:nonpauli} produces long-range nonstabilizerness, even for arbitrary connectivity~\cite{wei2025long}. This excludes $U_{\rm targ}=U_{\rm CD}C U_{\rm circ}$.
\end{proof}

\section{Non-additive codes}

\subsection{Generalities}
Given a stabilizer code $\mathcal{C}$ encoding $k$ logical qubits into $n$ physical qubits (denoted $[[n,k]]$), generated by $\mathcal S(\mathcal{C}) = \langle g_1, g_2 \dots g_{n - k}\rangle$, we will define the syndrome of a Pauli string $P$ as the vector $s(P) \in \{0, 1\}^{n-k}$ satisfying
\begin{align}
s_i(P) = \begin{cases}
0 & \text{if } [P, g_i] = 0,\\
1 & \text{if } \{P, g_i\}= 0.
\end{cases}
\end{align}
Any Pauli string with nonzero syndrome, $s(P)\neq 0^{n-k}$, is a detectable error.
\begin{definition}
Given a stabilizer code $\mathcal{C}$ and a set of Pauli errors with nonzero syndrome $E_1, E_2 \dots E_{R - 1}$ such that $s(E_i) \neq s(E_j)$, we define the corresponding non-additive code with
\begin{align}
\mathcal{C}' = \sum_{a = 0}^{R -1} E_a \mathcal{C},
\end{align}
where $E_0 = \1$.
\end{definition}
Since if $s(E_a) \neq s(E_b)$, $E_a \mathcal{C} \perp E_b \mathcal{C}$. Therefore, $\text{dim}(\mathcal{C}')= 2^{k} R$. We will typically assume that $R = 2^r$ for some $r \in \mathbb{Z}_{\geq 0}$. Then $\mathcal{C}'$ is a $[[n, k + r]]$ code.

The next lemma characterizes the set of detectable errors on $\mathcal{C}'$.
\begin{lemma}\label{lemma:det_error}
A Pauli string $E$ is a detectable error if either
\begin{enumerate}
\item $E \in \{\pm 1, \pm i\} \cdot \mathcal S(\mathcal{C})$ and  $[E, E_a] = 0$ $\forall a \in \{0, 1, 2 \dots R - 1\}$, or
\item $s(E) \notin \{s(E_a) + s(E_b): a, b \in \{0, 1, 2 \dots R - 1\}\}$.
\end{enumerate}
\end{lemma}
\begin{proof}
Note that
\begin{align}
\mathcal P_{\mathcal{C}'} = \sum_{a} E_a \mathcal P_{\mathcal{C}} E_a.
\end{align}
Now, suppose $E$ is a Pauli and $\mathcal P_{\mathcal{C}'}  E \mathcal P_{\mathcal{C}'} = c_E \mathcal P_{\mathcal{C}'} $ for some $c_E\in \mathbb{C}$. Then
\begin{align}
\sum_{a, b}E_a \mathcal P_{\mathcal{C}} E_a EE_b \mathcal P_{\mathcal{C}} E_b = c_E \sum_a E_a \mathcal P_{\mathcal{C}} E_a.
\end{align}
Using the fact that $E_a \mathcal{C} \perp E_b\mathcal{C}$ for $a \neq b$, we then obtain that
\begin{align}
\text{For } a\neq b, \;  \mathcal P_{\mathcal{C}} E_a E E_b \mathcal P_{\mathcal{C}} = 0 \quad \text{and} \quad  \mathcal P_{\mathcal{C}} E_a E E_a  \mathcal P_{\mathcal{C}}  = c_E  \mathcal P_{\mathcal{C}}.
\end{align}
Thus, we obtain that for $ a\neq b$,
\begin{align}
s(E_a E E_b) \neq 0 \implies s(E) \neq s(E_a) + s(E_b) .
\end{align}
Furthermore, since $\mathcal P_{\mathcal{C}}E_a E E_a \mathcal P_{\mathcal{C}} = (-1)^{\chi(E, E_a)} \mathcal P_{\mathcal{C}} E \mathcal P_{\mathcal{C}}$ where where $\chi(E, E_a) = 0$ if $[E, E_a] = 0$ and otherwise $1$, we obtain that 
\begin{enumerate}
\item $\mathcal P_{\mathcal{C}} E\mathcal P_{\mathcal{C}} = 0$ or equivalently $s(E) \neq 0$, in which case $c_E = 0$.
\item If $\mathcal P_{\mathcal{C}} E\mathcal P_{\mathcal{C}} \neq 0$, then $E$ must be a stabilizer of the code $\mathcal{C}$ and $\chi(E, E_a) $ is independent of $a\implies [E, E_a] = 0$. 
\end{enumerate}
\end{proof}

\subsection{CSS codes}
We next derive a set of sufficient conditions under which a product of single-qubit powers of the phase gate $S = \diag(1,i)$ is enough to correct the outcome of a measurement. We call this correction \emph{transversal}, since it acts independently on each physical qubit. Here we consider the restrictive setting of a CSS code~\cite{nielsen2010quantum}. Recall that a $[[n, k]]$ CSS code $\mathcal{C} = \text{CSS}(\mathcal{C}_X, \mathcal{C}_Z)$ is generated with two classical codes: an $X$ code $\mathcal{C}_X$ which is $[n, k_X]$ and a $Z$ code $\mathcal{C}_Z$ which is $[n, k_Z]$ and $k = k_X + k_Z - n$.  The generators of a CSS code are $\{h_X: h_X \in \mathcal{C}_X^\perp\} \cup \{h_Z: h_Z \in \mathcal{C}_Z^\perp\}$. We will only consider $X$ errors: For a binary column vector $e \in \mathbb F_2^n$, write $X^e = \bigotimes_{j=1}^n X^{e_j}$. Let $H_{Z} \in \mathbb F_2^{(n-k_Z)\times n}$ be a parity-check matrix of $\mathcal{C}_Z$, whose rows form a basis of $\mathcal{C}_Z^\perp$ and specify the $Z$-type stabilizer generators. The $Z$-check part of the syndrome of $X^e$ is $s(e) = H_{Z} e$, with arithmetic modulo 2; the $X$-check components of the full syndrome $s(X^e)$ vanish since $X^e$ commutes with all $X$ stabilizers. Furthermore, an $X$ Pauli $X^e$ is a detectable error on $\mathcal{C}$ if $e  \in \mathcal{C}_X^\perp \cup \mathcal{C}_Z^c$.

Given $e_1, e_2 \dots e_{R -1}$ such that $H_{Z} e_a \neq H_{Z} e_b \neq 0$, we can then define the non-additive code
\begin{align}
\mathcal{D} = \sum_{a = 0}^{R - 1} X^{e_a} \mathcal{C}.
\end{align}
A computational basis for CSS code can be constructed starting from $c_z \in \mathcal{C}_Z$, and from it constructing the state~\cite{gottesman2026surviving}
\begin{align}
\ket{\phi_{c_z}} =\frac{1}{\sqrt{\abs{\mathcal{C}_X^\perp}}} \sum_{h_x \in \mathcal{C}_X^\perp} \ket{c_z + h_x}.
\end{align}
Note that, if $c_z, c_z' \in \mathcal{C}_Z$ such that $c_z + c_z' \in \mathcal{C}_X^\perp$, then $\ket{\phi_{c_z}} = \ket{\phi_{c_z'}}$, and if $c_z + c_z' \notin \mathcal{C}_X^\perp$ then $\bra{\phi_{c_z}}\phi_{c_z'}\rangle = 0$. 

For a quantum state $\ket{\psi}$, let us define $\text{supp}_c(\ket{\psi})$ to be its computational basis support, i.e., a list of all the computational basis bit strings appearing in $\ket{\psi}$. Then, one can note that
\begin{align}
\text{supp}_c(\ket{\phi_{c_z}}) = \mathcal{C}_X^\perp + c_z.
\end{align}
We can extend this to construct a computational basis for $\mathcal{C}'$ by simply constructing a computational basis $\ket{\phi_{a, c_z}}$ for each of the subspaces $X^{e_a}\mathcal{C}_Z$
\begin{align}
\ket{\phi_{a, c_z}} = \frac{1}{\sqrt{\abs{\mathcal{C}_X^\perp}}} \sum_{h_x \in \mathcal{C}_X^\perp} \ket{e_a + c_z + h_x}.
\end{align}
We now show that distinct basis vectors of this form have entirely different computational basis support, i.e., $\forall a, a' \in \{0, 1, 2 \dots R - 1\}\text{ with } e_a \neq e_{a'}, c_z, c_z' \in \mathcal{C}_Z \text{ with } c_z + c_z' \notin \mathcal{C}_X^\perp: $
\begin{align}
\text{supp}_c(\ket{\phi_{a, c_z}}) \cap \text{supp}_c(\ket{\phi_{a', c_z'}}) = \emptyset.
\end{align}
To see this, assume it was not true: then, $\exists h_x, h_x' \in \mathcal{C}_X^\perp$ such that $c_z + e_a +h_x = c_z' + e_{a'} + h_{x}'$, from which it would follow that $H_{Z} e_a = H_{Z} e_{a'}$ since $H_{Z} c_z =H_{Z} c_{z}' = H_{Z} h_{x} = H_{Z} h_{x}'=  0$ which contradicts the assumption that $H_{Z} e_a \neq H_{Z} e_{a'}$. 

We now state the main result of this section.
\begin{lemma}\label{lemma:Det_error_CSS}
Suppose that $\mathcal{C}'$ is the code as constructed above, and let us measure $X^{f_1}, X^{f_2} \dots X^{f_M}$ on the code-space. Then, this constructs a CMF protocol correctable with transversal $S$ gates if and only if
\begin{enumerate}
\item[(a)] The operation $\prod_{i = 1}^M (\1 + X^{f_i})$ is proportional to an isometry on $\mathcal{D}$ if $\forall f \in \textnormal{span}(\{f_i: i \in\{1, 2 \dots M\}\}) \subseteq \mathcal{C}_X^\perp $ or $H_{Z} f \notin \{H_{Z} (e_a + e_b) : \forall a, b \in \{0, 1 \dots R-1\}\}$.
\item[(b)]For every $\boldsymbol a \in \{0, 1\}^M$, $\exists y \in \{0, 1, 2, 3\}^n$ dependent only on $\boldsymbol a$ such that $\forall v \in \{0, 1\}^M$ and $z \in \bigcup_{a = 0}^{R - 1}(e_a + \mathcal{C}_Z)$:
\begin{align}
y^\text{T}\bigg(z + \sum_{i = 1}^M v_i f_i\bigg) = 2 \sum_{i = 1}^M v_i a_i  \textnormal{ mod }4.
\end{align}
\end{enumerate}
\end{lemma}
\begin{proof}
(a) It follows from the general \cref{lemma:det_error} applied to $X^{f}$ for $f \in \text{span}(\{f_i: i \in\{1, 2 \dots M\}\})$.

(b)  It is enough to show this for a basis for $\mathcal{D}$. Pick $\ket{\psi} = \ket{\phi_{a, c_z}}$. Note that
\begin{align}
\prod_{i = 1}^M \bigg(\frac{\1 + (-1)^{a_i }X^{f_i}}{2}\bigg) = \frac{1}{2^M} \sum_{v \in \{0, 1\}^M} (-1)^{\sum_{i = 1}^M v_i a_i } X^{\sum_{i = 1}^M v_i  f_i},
\end{align}
and thus it is enough to ensure that $\forall \boldsymbol a , v \in \{0, 1\}^M, a \in \{0, 1 \dots R - 1\}, c_z \in \mathcal{C}_Z$ there is a $y \in \{0,1,2,3\}^n$ dependent only on $\boldsymbol a$ such that
\begin{align}\label{eq:condition_s_gate}
S^y  (-1)^{\sum_{i = 1}^M v_i a_i}X^{\sum_{i = 1}^M v_i f_i}  \ket{\phi_{a, c_z}} = X^{\sum_{i = 1}^M v_i \ f_i}  \ket{\phi_{a, c_z}} .
\end{align}
Using the decomposition of $\ket{\phi_{a, c_z}}$ onto the computational basis, we can write
\begin{align}
X^{\sum_{i = 1}^M a_i v_i f_i}\ket{\phi_{a, c_z}} = \frac{1}{\sqrt{\abs{\mathcal{C}_X^\perp}}}\sum_{h_x \in \mathcal{C}_X^\perp} \ket{e_a + c_z + h_x + \sum_{i = 1}^M v_i f_i}, 
\end{align}
and
\begin{align}
S^y X^{\sum_{i = 1}^M  v_i f_i}\ket{\phi_{a, c_z}} = \frac{1}{\sqrt{\abs{\mathcal{C}_X^\perp}}}\sum_{h_x \in \mathcal{C}_X^\perp} i^{y^\text{T}(e_a + c_z + h_x + \sum_{i = 1}^M v_i f_i)} \ket{e_a + c_z + h_x + \sum_{i = 1}^M v_i f_i}.
\end{align}
Therefore, the condition in Eq.~\eqref{eq:condition_s_gate} can be simplified to
\begin{align}
y^\text{T}\bigg(z + \sum_{i = 1}^M  v_i f_i \bigg) + 2\sum_{i = 1}^M v_i a_i = 0 \text{ mod }4 \ \forall z \in \bigcup_{a = 0}^{R - 1}\big(e_a + \mathcal{C}_Z\big),
\end{align}
which proves the lemma.
\end{proof}
A simple sufficient condition is as follows.
\begin{lemma}
Suppose for every $i \in \{1,\dots,M\}$, $\exists y_i \in \{0, 1, 2, 3\}^n$ such that $\forall z \in \bigcup_{a = 0}^{R - 1}  (e_a + \mathcal{C}_Z)$ and $v \in \{0, 1\}^M$
\begin{align}\label{eq:condition_simplified}
y_i^\textnormal{T} \bigg(z + \sum_{i = j}^M v_j f_j\bigg) = 2v_i \textnormal{ mod }4,
\end{align}
then both conditions of the previous lemma hold.
\end{lemma}
\begin{proof}  We check both the conditions.

\emph{Condition (a)}. Setting $v = 0$, we obtain that $y_i^\text{T} z = 0 \text{ mod }4$ for all $i \in \{1, 2 \dots M\}$ and $z \in \bigcup_{a = 0}^{R - 1} (e_a +  \mathcal{C}_Z)$. Now, suppose $\exists f(v') = \sum_{i = 1}^M v'_i f_i$ such that $f \notin \mathcal{C}_X^\perp$ and $H_{Z} f(v') = H_{Z}(e_a + e_b)$ for some $a, b \in\{0, 1, 2 \dots R - 1\}$. This is equivalent to $f(v') = e_a + e_b + c_z$ for some $c_z \in \mathcal{C}_Z$, and thus $e_a + f(v') = e_b + c_z \in  \bigcup_{a = 0}^{R - 1}  (e_a + \mathcal{C}_Z)$. Choose $i$ such that $v_i = 1$, then from Eq.~\eqref{eq:condition_simplified} using $z = e_a$, we obtain that
\begin{align}
y_i^\text{T}(e_a + f(v')) = 2v'_i \text{ mod }4 = 2.
\end{align}
Furthermore, since $e+ f(v') \in \bigcup_{a = 0}^{R - 1}  (e_a + \mathcal{C}_Z)$, also from Eq.~\eqref{eq:condition_simplified} choosing $z = e + f(v')$ and $v = 0$, we obtain that
\begin{align}
y_i^\text{T}(e_a + f(v')) = 0,
\end{align}
which is a contradiction.

\emph{Condition (b)}. Choose $y = \sum_{i = 1}^M a_i y_i$, then
\begin{align}
y^\text{T}\bigg(z + \sum_{j = 1}^M v_j f_j\bigg) = 2 \sum_{i = 1}^M a_i v_i  \text{ mod }4.
\end{align}
\end{proof}
\subsection{Example}
Consider the code $\mathcal{D} = \mathcal{D}_0^{\otimes m}$ on a total of $n = 6m$ qubits, where $\mathcal{D}_0 = \text{span}(\{x \in \{0, 1\}^6: \text{wt}(x) \in \{0, 4\})$. The code $\mathcal{D}_0$ can be seen as a non-additive code:
\begin{align}
\mathcal{D}_0 = \mathcal{C}_0 + \sum_{e \in \{0, 1\}^6 : \text{wt}(e) = 4} X^e \mathcal{C}_0 \quad \text{where} \quad \mathcal{C}_0 = \text{span}(\{0^6\}).
\end{align}
$\mathcal{C}_0$ is a stabilizer code generated by $Z_1, Z_2 \dots Z_6$ and the syndrome $X^e$ on $\mathcal{C}_0$ is simply $e$. Now, we consider the following measurement operators:
\begin{align}
P_i = \tilde{X}_i \tilde{X}_{i + 1} \quad \text{for} \quad i \in \{1, 2 \dots m - 1\}\quad \text{where}\quad \tilde{X}_i =\prod_{j =1}^6 X_{6(i - 1) + j},
\end{align}
i.e., $\tilde{X}_i$ is a product of $X$ operators on the $i^\text{th}$ code block. We note that any products of $P_i$ are detectable errors, either by verifying the condition in Lemma~\ref{lemma:Det_error_CSS}, or by simply noting that all $X$ on a code block maps bit strings of weight 0, 4 to bit strings of weight 6, 2. Therefore, measuring $P_i$ implements a CMF isometry.

Furthermore, this permits a transversal $S$ correction unitary. To see this directly, we note that since $\mathcal{D}$ contains only states with $0 \text{ mod } 4$ 1s,
\begin{align}
S^{\otimes 6} \mathcal{D}_0 = \mathcal{D}_0 \quad \text{and} \quad Z^{\otimes 6}\mathcal{D}_0= \mathcal{D}_0.
\end{align}
Now, for the measurement outcome $a_1, a_2 \dots a_{m - 1}\in\{0, 1\}$, pick $v_1, v_2 \dots v_m \in \{0, 1\}$ such that $ v_i \oplus v_{i + 1} = a_i$, and define the blockwise correction $V_{\boldsymbol a} = \bigotimes_{i=1}^{m}(S^{v_i})^{\otimes 6}$, which applies $S^{v_i}$ to each of the six qubits in block $i$. Using that $S X S^\dagger= iXZ$, we obtain
\begin{align}
V_{\boldsymbol a} \prod_{i = 1}^{m - 1} \bigg(\frac{\1 + (-1)^{a_i}\tilde{X}_i \tilde{X}_{i + 1}}{2}\bigg) \mathcal{D} &=  \prod_{i = 1}^{m - 1} \bigg(\frac{\1 + (-1)^{a_i}i^{6(v_i \oplus v_{i + 1})}\tilde{X}_i \tilde{X}_{i + 1}\tilde{Z}_i \tilde{Z}_{i + 1}}{2}\bigg) \mathcal{D}\nonumber \\
&= \prod_{i = 1}^{m - 1} \bigg(\frac{\1 + \tilde{X}_i \tilde{X}_{i + 1}}{2}\bigg) \mathcal{D},
\end{align}
where $\tilde Z_i$ is defined similarly to $\tilde X_i$ and in the last step we have used the fact that $[\tilde{Z}_i, \tilde{X}_j] = 0$ since each code block has an even number of qubits.

Next, we show that this isometry can generate long-range entangled states. The argument is very similar to how GHZ state can be created but treating the individual code block as one qubit: Start from the state $\ket{0}^{\otimes 6m} \in \mathcal{D}^{\otimes m}$, then
\begin{align}
\prod_{i = 1}^{m - 1}\bigg(\frac{\1 + \tilde{X}_i \tilde{X}_{i + 1}}{\sqrt{2}}\bigg)  \ket{0^{\otimes 6m}} &= \frac{1}{2^{(m -1)/2}} \sum_{b\in\{0, 1\}^m: \abs{b} \text{ is even}} \ket{b_1}^{\otimes 6}\ket{b_2}^{\otimes 6} \dots \ket{b_m}^{\otimes 6} \nonumber \\
&= \frac{1}{\sqrt{2}}
\bigl(\ket{\tilde{+}}^{\otimes m}
+\ket{\tilde{-}}^{\otimes m}\bigr),
\end{align}
where $\ket{\tilde{\pm}} = (\ket{0}^{\otimes 6} \pm \ket{1}^{\otimes 6})/\sqrt{2}$.

Finally, we show that this isometry cannot be implemented by Clifford, i.e., there cannot exist a $n = 6m$ qubit Clifford circuit $C$ such that
\begin{align}
\prod_{i = 1}^{m - 1} \bigg(\frac{\1 + \tilde{X}_i \tilde{X}_{i + 1}}{\sqrt{2}}\bigg) \ket{\psi} =  C \ket{\psi} \ \forall \ \ket{\psi} \in \mathcal{D}.
\end{align}
For this, we simply note that the only independent Paulis that stabilize $\mathcal{D}$ are $\tilde{Z}_1, \tilde{Z}_2 \dots \tilde{Z}_m$ yielding a stabilizer rank of $m$. However, the space $\prod_{i = 1}^{m - 1} \big(\frac{\1 + \tilde{X}_i \tilde{X}_{i + 1}}{\sqrt{2}}\big)\mathcal{D}$ is stabilized by $\tilde{Z}_1, \tilde{Z}_2 \dots \tilde{Z}_m$ as well as $\tilde{X}_1\tilde{X}_2, \tilde{X}_2\tilde{X}_3 \dots \tilde{X}_{m - 1}\tilde{X}_m$, yielding a total of $2m - 1$. Since a Clifford unitary cannot change the stabilizer rank of a subspace, $C$ cannot exist.

\end{document}